\documentclass[a4paper]{article}

\usepackage{amsmath}
\usepackage[psamsfonts]{amssymb}
\usepackage{latexsym}
\usepackage{cmmib57}
\usepackage{amscd}
\usepackage{graphicx}
\usepackage{pstricks}
\usepackage{pst-coil}
\usepackage{pst-node}
\usepackage{float}
\usepackage{wrapfig}
\usepackage{epsf}

\newgray{meascolor}{.5}
\newgray{opercolor}{.5}
\newgray{entacolor}{.5}
\newgray{qbitcolor}{.1}
\newgray{mklight}{.7}
\newgray{mkfill}{.9}
\psset{unit=\unitlength}
\psset{arrowsize=2pt 5}
\psset{coilarm=.4}

\def\idty{{\leavevmode{\rm 1\ifmmode\mkern -4.8mu\else\kern -.3em\fi
      I}}}
\renewcommand{\Bbb}[1]{\if1#1\idty\else\mathbb{#1}\fi}

\newcommand{\kb}[1]{|#1\rangle\langle#1|}
\newcommand{\KB}[2]{|#1\rangle\langle#2|}
\newcommand{\ket}[1]{|#1\rangle}

\newcommand{\U}{\operatorname{U}}

\newtheorem{thm}{Theorem}[section]
\newtheorem{defi}[thm]{Definition}
\newtheorem{prop}[thm]{Proposition}

\newtheorem{kor}[thm]{Corollary}

\newenvironment{proof}{\par\noindent\textit{Proof.\ }}{\hfill $\Box$ \vspace{1em}}

\newtheorem{aX}{Axiom} 

\usepackage[utf8]{inputenc}
\usepackage[giveninits=true,style=alphabetic,maxbibnames=99,sorting=nyt]{biblatex}
\usepackage{enumitem}
\usepackage{geometry}
\usepackage{tabularx}
\usepackage{wrapfig}

\bibliography{dfg_math_18.bib}

\title{Quantum control in infinite dimensions and Banach-Lie algebras: Pure point spectrum}

\author{Michael Keyl\\[1em]
{\small Dahlem Center for Complex Quantum Systems,}\\ {\small Freie Universität Berlin, 14195 Berlin, Germany}\\[1em]
{\small \texttt{michael.keyl@tum.de}}
}

\begin{document}

\maketitle

\begin{abstract}
  In finite dimensions, controllability of bilinear quantum control systems can be decided quite easily in terms of the ``Lie algebra rank condition'' (LARC), such that only the systems Lie algebra has
  to be determined from a set of generators. In this paper we study how this idea can be lifted to infinite dimensions. To this end we look at control systems on an infinite dimensional Hilbert space
  which are given by an unbounded drift Hamiltonian $H_0$ and bounded control Hamiltonians $H_1, \dots, H_N$. The drift $H_0$ is assumed to have empty continuous spectrum. We use recurrence methods
  and the theory of Abelian von Neumann algebras to develop a scheme, which allows us to use an approximate version of LARC, in order to check approximate controllability of the control system in
  question. Its power is demonstrated by looking at some examples. We recover in particular previous genericity results with a much easier proof. Finally several possible generalizations are
  outlined.  
\end{abstract}

\section{Introduction}
\label{sec:introduction}

One of the most fundamental, mathematical questions in quantum control theory is controllability: can one reach a particular target state, or implement a particular unitary gate (unitary operator) by
manipulating a given set of control vector fields (controls Hamiltonians). For closed quantum systems modeled on {\em finite} dimensional Hilbert spaces this question can be answered in a very
systematic way using appropriately assigned Lie subalgebras and the Lie-algebra rank condition (LARC) \cite{sussmann1972controllability,jurdjevic1972control,brockett1972system,
brockett1973lie,dirrhelmke2008}. The only mathematical desideratum are efficient ways to determine a Lie algebra from its generators. In this area, symmetry based
techniques led to quite complete solutions for spin systems and lattice Fermions \cite{zeier2011symmetry,zimboras2014dynamic,zimboras2015symmetry,zeier2015squares}. 

Moving from finite to infinite dimensions makes the situation substantially more difficult. One may compare it to changing from ordinary to partial differential equations. A fundamental observation
is that (in general) \emph{exact} controllability becomes impossible and has to be replaced by an \emph{approximate} version \cite{ball1982controllability}. In other words, states or unitary
operators can usually not be reached exactly but, if at all, only approximately (within a given topology). Even with this generalization challenging difficulties remain.  

One way of addressing them is to take the PDE picture seriously by casting the Schrödinger equation into the standard framework of PDE control and asking whether a wave function is approximately
reachable from a given initial state. Virtually all studies on infinite dimensional quantum systems treat the controllability problem within this wave function picture. Among the methods used are
adiabatic evolution \cite{adami2005controllability,boscain2012adiabatic,boscain2015approximate}, Lyapunov methods
\cite{mirrahimi2005lyapunov,nersesyan2009growth,nersesyan2010global,nersesyan2012global,%
  morancey2014global}, applications of the Nash-Moser implicit function theorem \cite{beauchard2005local,beauchard2006controllability}, 
and finally Lie algebraic methods in connection with finite dimensional subspaces which are either invariant
\cite{brockett2003controllability,rangan2004control,yuan2007controllability,bloch2010finite,KZSH,HoKe17,HeKe18} 
or, in an appropriate way, approximately invariant
\cite{chambrion2009controllability,boscain2012weak,chambrion2012periodic,boscain2014multi,boscain2015control,paduro2015approximate,paduro2015control,%
  chitour2016generic}. While most of the work is dedicated to systems the Hamiltonians of which possess a discrete spectrum, 
one of the few exceptions is \cite{beauchard2010controllability}.
  
Yet there is an alternative point of view: we can look at the outlined control problem as part of {\em operator theory}, where we are no longer interested in states (represented by wave functions or
density operators) but rather in the operator lift in terms of the structure of the underlying groups (and semigroups) generated by the unitary operators one can reach approximately from the identity
operator $\Bbb{1}$ . Closely related is again a Lie algebraic picture, now in terms of Banach-Lie algebras consisting of bounded operators and being generated -- possibly in an indirect way -- by
the Hamiltonians of the system.  The central task is to establish an appropriate generalization of the LARC and to discuss the 
limitations, arising from infinite dimensions.

The purpose of this paper is to make first steps into this direction. To this end we present in Section \ref{sec:setup} a general setup, together with a candidate for an approximate version of LARC in
infinite dimensions. There are two technical problems connected to the implementation of this new definition (cf. Problem 1 and 2 in Sect. \ref{sec:setup} below). In general they block a
straightforward application of the new condition within controllability proofs. In the special case of a control system where only the drift Hamiltonian $H_0$ can be unbounded (and the control
Hamiltonians $H_1, \dots H_N$ are bounded), with a spectrum consisting only of (not necessarily isolated) eigenvalues we will show that both problems can be avoided or circumvented. This will happen
in Sects.\ref{sec:recurrence} and \ref{sec:maximal-torus}. Crucial roles in this context are played by recurrence methods (Sect. \ref{sec:recurrence}) and the theory of Abelian von Neumann algebras
(Sect. \ref{sec:maximal-torus}). Within the described limitations of the systems under consideration the new techniques allow a wide range of approximate controllability proof just by calculating
commutators between bounded operators. The only two ingredients which are new -- compared  to finite dimensions -- are strong convergence arguments (which are usually easy to handle), and some 
spectral analysis of the unbounded drift  $H_0$ (depending on $H_0$ this can of course be difficult). In Sect. \ref{sec:contr-theor} the new scheme is demonstrated by recovering earlier results from
\cite{boscain2012weak,boscain2014multi,caponigro2018exact} with a much easier proof. For systems where our assumptions fail (e.g. if some of the control Hamiltonian are unbounded, too), generalizations
of the proposed schemes are necessary. Some possible directions are outlines in Sect. \ref{sec:generalizations}. The paper closes with an outlook in Sect. \ref{sec:outlook}.   

\section{Setup}
\label{sec:setup}

Let us start with a separable Hilbert space $\mathcal{H}$, the corresponding von Neumann algebra $\mathcal{B}(\mathcal{H})$ of bounded operators, the group
$\mathrm{U}(\mathcal{H})$ of unitary elements of $\mathcal{B}(\mathcal{H})$, and its Lie algebra $\mathfrak{u}(\mathcal{H})$ consisting of bounded, anti-selfadjoint
operators. We equip all spaces with the strong topology, and this turns $\mathrm{U}(\mathcal{H})$ into a topological group \cite{schottenloher2018unitary}. 
The central object we are interested in is the time-dependent Schrödinger equation  
\begin{equation} \label{eq:1}
  \frac{d}{dt} U(t,s) = -i H(t) U(t,s), \quad H(t) = H_0 + \sum_{j=1}^N u_j(t) H_j,
\end{equation}
where $\Bbb{R}^2 \ni (t,s) \mapsto U(t,s) \in \mathrm{U(\mathcal{H})}$ is a unitary propagator, the $H_j$, $j=0,\dots,N$ are selfadjoint (possibly unbounded) operators on
$\mathcal{H}$, the $u_j : \Bbb{R} \rightarrow \Bbb{R}$, $j=1,\dots, N$ are piecewise constant control functions, and the time differential has to be understood in the
strong sense. For $j=1,\dots,N$ the $H_j$ are called \emph{control Hamiltonians}, while $H_0$ is denoted the \emph{drift}. Furthermore, the following two assumptions are
required
\begin{enumerate}[leftmargin=*]
\item
  The operators $H_j$, $j=0,\dots,N$ admit a joint dense domain $D$, and on $D$ all linear combinations
  \begin{equation} \label{eq:4}
    H({\mathbf{y}}) = H_0 + \sum_{j=1}^N y_j H_j, \quad \mathbf{y} = (y_1,\dots,y_N) \in \Bbb{R}^N,
  \end{equation}
  are essentially selfadjoint. By the Kato-Rellich Theorem \cite[Thm X.12]{RESI2} this requirement is automatically satisfied, if all control Hamiltonians $H_j$, $j\geq
  1$ are bounded. In that case $D$ can be chosen as the domain of $H_0$, and all $H({\mathbf{y}})$ are selfadjoint on it (i.e. not just \emph{essentially}
  selfadjoint). Boundedness of the $H_j$ for $j>0$ is an important assumption in Sect. \ref{sec:recurrence} and \ref{sec:contr-theor}.   
\item
  The ranges $u_j(\Bbb{R})$ of the functions $u_j$, $j=1,\dots,N$ are finite sets, and the inverse images $u_j^{-1}(y)$ are for all $y \in \Bbb{R}$ and all $j=1,\dots,N$
  either empty or a union of finitely many intervals. Furthermore, we assume that the supports of all $u_j$ are contained in an interval $[0,T]$. The number $T \in
  \Bbb{R}$ is called the \emph{control time.}
\end{enumerate}
With the given assumptions it is easy to see that the propagator $U$ exists and is unique. It is even possible to express it in terms of exponentials. To see the latter
note first that we can find a partition $a = t_0 < t_1 < \dots < t_M = b$ of the interval $[a,b]$ such that all $u_j$ are constant on the subintervals
$(t_{k-1},t_k]$, $k=1, \dots, M$. The latter property is shared by the Hamiltonians $H(t)$ and therefore we get
\begin{equation} \label{eq:2}
  U(a,b) = \exp\bigl(i\tau_1 H(\mathbf{y}_1)\bigr) \dots \exp\bigl(i\tau_M H(\mathbf{y}_M)\bigr),\quad \tau_k = t_k - t_{k-1},\quad
  \mathbf{y}_k = \bigl(u_1(t_k), \dots, u_N(t_k)\bigr)\,.
\end{equation}
Also note that outside of the interval $[0,T]$ the time evolution is given alone by the drift, i.e. $U(a,0) = \exp\bigl(iaH_0\bigr)$ and $U(b,T) =
\exp\bigl(i(b-T)H_0\bigr)$ for $a < 0$ and $b > T$. This defines the propagators $U(a,b)$ explicitly for all $a,b \in \Bbb{R}$, and we can use them together with Eq.  (\ref{eq:2}) to introduce a
number of additional objects. 

\begin{defi} \label{def:1}
  On a separable Hilbert space $\mathcal{H}$ consider selfadjoint operators $H_0,\dots,H_N$ such that all the linear combination $H(\mathbf{y})$ from Eq. (\ref{eq:4}) are
  essentially selfadjoint on a common dense domain $D$. We define
  \begin{enumerate}[leftmargin=*]
  \item
    $\mathcal{R}_0(H_0;H_1,\dots,H_N)$ as the smallest subsemigroup of $\mathrm{U}(\mathcal{H})$ containing all unitaries $\exp\bigl(i \tau H(\mathbf{y})\bigr)$ with
    $\tau > 0$ and $\mathbf{y} \in \Bbb{R}^N$. It is called the \emph{algebraically reachable set.}
  \item
    The strong closure of $\mathcal{R}_0(H_0;H_1,\dots,H_N)$  is denoted by $\mathcal{R}(H_0;H_1, \dots, H_N)$ and called the \emph{strongly reachable set}.
  \item
    Similarly we define the \emph{dynamical group} $\mathcal{G}(H_0, \dots, H_1)$ as the smallest, strongly closed subgroup of $\mathrm{U}(\mathcal{H})$ containing all
    unitaries $\exp(i \tau H(\mathbf{y}))$ with $\tau \in \Bbb{R}$ and $\mathbf{y} \in \Bbb{R}^N$.
    \emph{dynamical group}. 
  \item
    The \emph{restricted dynamical group} $\hat{\mathcal{G}}(H_0, \dots, H_N)$ is the smallest, strongly closed subgroup of $\mathrm{U}(\mathcal{H})$ containing all
    $\exp(itH_j)$ for all $t \in \Bbb{R}$ and all $j=0,\dots,N$.
  \end{enumerate}
  If confusion can be avoided, we frequently drop the arguments and write $\mathcal{R}_0$, $\mathcal{R}$, $\mathcal{G}$ and $\hat{\mathcal{G}}$.
\end{defi}

\noindent Using the reasoning of Eq. (\ref{eq:2}) we see that $\mathcal{R}_0$ can be defined alternatively as the set of all unitaries arising as propagators $U(T,0)$ for a certain
control time and with appropriate control functions. This is the reason for the name ``reachable set''. It contains all unitaries which can be reached by following the
time evolution (\ref{eq:1}) from the initial condition $U(0,0) = \Bbb{1}$ up to time $T$.

The relationship between $\mathcal{G}$ and $\hat{\mathcal{G}}$ is a bit tricky. It is unclear whether one of the inclusions $\mathcal{G} \subset \hat{\mathcal{G}}$ or
$\hat{\mathcal{G}} \subset \mathcal{G}$ or even equality holds. The conjecture for all three cases is: ``no''. Promising sources for counterexamples are pairs of
unbounded operators $A$, $B$ with a commutator $[A,B]$ which vanishes on a joint core of $A,B$, but with non-commuting time evolutions $\exp(itA)$, $\exp(itB)$;
cf. \cite[Sect. VIII.5]{RESI1}. An important difference between $\hat{\mathcal{G}}$ and $\mathcal{G}$ is the scope of the definition: For $\hat{\mathcal{G}}$ to be well
defined selfadjointness of the $H_j$, $j=0,\dots,N$ is completely sufficient. We do not need the additional assumption on the linear combinations $H(\mathbf{y})$ which
are required for $\mathcal{G}$. Fortunately, with a sufficiently ``regular'' setup the difference $\mathcal{G}$ and $\hat{\mathcal{G}}$ vanishes.

\begin{prop} \label{prop:1}
  Assume that all the operators $H(\mathbf{y})$, $y \in \Bbb{R}^N$ from Eq. (\ref{eq:4}) are defined and selfadjoint on the same domain $D$. Then $\hat{\mathcal{G}} =
  \mathcal{G}$ holds. 
\end{prop}

\begin{proof}
  To prove $\mathcal{G} \subset \hat{\mathcal{G}}$, we have to show that for all $\mathbf{y} = (y_1,\dots,y_N) \in \Bbb{R}^N$ the operators
  \begin{equation}
    \exp\bigl(itH(\mathbf{y})\bigr) = \exp(itH_0 + ity_1H_1 + \dots + ity_N H_N)
  \end{equation}
  can appear as a strong limit of products $\exp(it_1H_1) \dots \exp(itH_N)$. However, this follows easily from Trotter's formula \cite{RESI1}[Thm. VIII.30], which we can
  apply since all involved operators are defined and selfadjoint on the same domain. For the other inclusion $\hat{\mathcal{G}} \subset \mathcal{G}$ note first that
  $\exp(itH_0) \in \mathcal{G}$ always holds, since $H_0 = H(\mathbf{y})$ with $\mathbf{y} = 0$. It remains to show that we have $\exp(itH_j) \in \mathcal{G}$ for
  $j=1,\dots,N$, too. But due to $H_j = H_{e_j} - H_0$ with the $j^{\mathrm{th}}$ element $e_j$ of the canonical basis $e_1, \dots, e_N \in \Bbb{R}^N$, we can again
  use a Trotter argument.
\end{proof}

\noindent If we look at weaker versions of Trotter's equation (e.g. \cite[Thm. VIII.31]{RESI1}) we see that there is room for generalizations. However, the necessary conditions on
the relationships of the domains of the $H(\mathbf{y})$ are cumbersome (to say the least). On the other hand, in this paper we are mostly interested in models where the
control Hamiltonians $H_j$, $j=1,\dots,N$ are bounded. In that case the assumptions of the proposition are met; cf. the corresponding remarks above.

To proceed we need strong continuity of the exponential map $\exp : \mathfrak{u}(\mathcal{H}) \rightarrow \mathrm{U}(\mathcal{H})$. Although, this seems to be a well
known fact, a published proof seems to be unavailable. Therefore we provide one. The idea was suggested by Neeb \cite{neebPC}.

\begin{prop}
  The exponential map  $\exp: \mathfrak{u}(\mathcal{H}) \rightarrow   \mathrm{U}(\mathcal{H})$ is continuous with respect to the strong topologies on
  $\mathfrak{u}(\mathcal{H})$ and $\mathrm{U}(\mathcal{H})$. 
\end{prop}

\begin{proof}
  On $\mathfrak{u}(\mathcal{H})$ the strong topology is not first countable. Therefore convergent sequences are not sufficient to prove continuity. We
  have to use nets, instead. Hence consider an directed set $(I, \leq)$ and over it a net $H_\lambda$, $\lambda \in I$ of bounded selfadjoint
  operators converging strongly to $H \in \mathcal{B}(\mathcal{H})$. Furthermore, we assume that $H$ and all $H_\lambda$ are selfadjoint, i.e. $iH_\lambda, iH \in \mathfrak{u}(\mathcal{H})$. 

  If the net $H_\lambda$ is in addition bounded, i.e. $\|H_\lambda\| \leq K$ for all $\lambda \in I$ we can show, with a simple induction argument that the net $P(H_\lambda)$ defined in terms of an  
  arbitrary polynomial $P$ in one variable is strongly convergent, too. All we need is the estimate (with $\psi \in \mathcal{H}$)
  \begin{equation}
    \| (H_\lambda^n - H^n) \psi\| \leq \| H_\lambda^{n-1}(H_\lambda - H)\psi\| + \| (H_\lambda^{n-1} - H^{n-1})H\psi\|.
  \end{equation}
  If $\| (H_\lambda -H)\psi\| < \epsilon/(2K)$ and $\|(H_\lambda^{n-1} - H^{n-1})H\psi\|< \epsilon/2$ holds, we get $\| (H_\lambda^n - H^n)\| <
  \epsilon$ and by induction on $n$ we see that $H_\lambda^n \rightarrow H^n$ strongly as claimed. By taking linear combinations we can extend the
  result to polynomials.

  Unfortunately, strongly convergent nets are (in contrast to strongly convergent sequences) not necessarily bounded. It is therefore convenient to
  look at resolvent operators. For all selfadjoint, (bounded) operators $X$ we have $\|X\pm i \Bbb{1}\| \geq 1$ hence $\|(H_\lambda \pm
  i\Bbb{1})^{-1}\| \leq 1$. In other words the nets of resolvents are bounded. To see that they are strongly convergent, too, look at
  \begin{equation}
    (H_\lambda \pm i\Bbb{1})^{-1} - (H \pm i\Bbb{1})^{-1} = (H_\lambda \pm i \Bbb{1})^{-1} (H - H_\lambda) (H \pm i\Bbb{1})^{-1}.
  \end{equation}
  By assumption we know that $(H - H_\lambda) (H \pm i\Bbb{1})^{-1}\psi \rightarrow 0$ holds for all $\psi \in \mathcal{H}$. Hence, since $\|(H_\lambda \pm
  i \Bbb{1})^{-1}\| \leq 1$ we get $(H_\lambda \pm i\Bbb{1})^{-1} \rightarrow (H \pm i\Bbb{1})^{-1}$ strongly as claimed.

  Now we can replace in the proof of Thm. VIII.20(b) in \cite{RESI1} the sequence by a strongly convergent net and all arguments remain valid (cf. in
  particular the above remark on polynomials). Therefore we can conclude that $\exp(iH_\lambda) \rightarrow \exp(iH)$ holds strongly, and this proves
  strong continuity of $\exp$.
\end{proof}

\noindent $\mathcal{G}$ is in general not a Lie group -- neither in the strong nor in the uniform topology. Nevertheless, we can associate a Lie algebra to it. This is done in the
next proposition.

\begin{prop}
  Consider the assumptions from Def. \ref{def:1} and in particular the dynamical group $\mathcal{G}$. The equation
  \begin{equation}
    \mathfrak{g}(H_0,\dots,H_N) = \{ X \in \mathfrak{u}(\mathcal{H})\,|\, \exp(tX) \in \mathfrak{g}(H_0,\dots,H_N),\ \forall t \in \Bbb{R}\},
  \end{equation}
  defines a strongly closed Lie subalgebra of $\mathfrak{u}(\mathcal{H})$, which is called the \emph{dynamical Lie algebra}. As with $\mathcal{G}$ we just write
  $\mathfrak{g}$ without arguments if confusion can be avoided.
\end{prop}

\begin{proof}
  Two bounded operators $A, B$ on $\mathcal{H}$ satisfy the following two equations
  \begin{gather}
    \lim_{n \rightarrow \infty} \left[\exp\left(\frac{A}{n}\right) \exp\left(\frac{B}{n}\right)\right]^n = \exp(A+B) \label{eq:5}   \\
    \lim_{n \rightarrow \infty} \left[\exp\left(\frac{A}{n}\right) \exp\left(\frac{B}{n}\right)\exp\left(-\frac{A}{n}\right) \label{eq:6}
      \exp\left(-\frac{B}{n}\right)\right]^{n^2} = \exp\bigl([A,B]\bigr) .
  \end{gather}
  Eq. (\ref{eq:5}) is the Trotter product formula \cite{RESI1}[Thm. VIII.30] already used in a more general setting in
  Prop. \ref{prop:1}. Eq. (\ref{eq:6}) is probably a bit less well known and called commutator formula. The infinite dimensional version presented here can
  be shown as in finite dimensions by using the Baker-Campbell-Hausdorff formula \cite[Prop. 3.4.7]{hilgert2011structure}. Now we use Eq. (\ref{eq:5})
  to show that $\mathfrak{g}$ is a linear subspace of $\mathfrak{u}(\mathcal{H})$ and Eq. (\ref{eq:6}) that it is even a Lie subalgebra. To show
  closedness in the strong topology we need the corresponding continuity of the exponential map. Consider a net $X_\lambda \in \mathfrak{g}$, $\lambda
  \in I$ which is strongly convergent to $X \in \mathfrak{u}(\mathcal{H})$. We  have to show that $X \in  \mathfrak{g}$. By assumption we have
  $\exp(tX_\lambda) \in \mathcal{G}$ for all $t$ and all $\lambda \in I$. Furthermore the scaled nets $tX_\lambda$ are still convergent with limit
  $tX$. By the continuity of $\exp$ we therefore get $\exp(tX) \in \mathcal{G}$ for all $t$ and the definition of $\mathfrak{g}$ shows that $X \in
  \mathfrak{g}$ holds. 
\end{proof}

\noindent We are now ready to state the main problem we want to solve with this paper. The following definition of strong controllability is taken form
\cite{KZSH} and \cite{HoKe17}, but is was introduced earlier in \cite{boscain2012weak} under the name: \emph{approximate simultaneous
  controllability}. 

\begin{defi}
  The control system (\ref{eq:1}) is called \emph{exactly operator controllable} if $\mathcal{R}_0 = \mathrm{U}(\mathcal{H})$ holds and \emph{strongly
    operator controllable} if $\mathcal{R}(H_0,H_1,\dots,H_N) = \mathrm{U}(\mathcal{H})$ is satisfied.
\end{defi}

\noindent The idea behind this definition is that all unitaries on $\mathcal{H}$ can appear as (strong approximations of) solutions $U(0,T)$ to the differential
equation (\ref{eq:1}) with initial data $U(0,0)=\Bbb{1}$ and appropriately chosen control functions. In the language of physics, we can rephrase this as: all unitary ``gates'' can be realized in a
good approximation (with respect to the strong topology) by controlling the Schrödinger equation (\ref{eq:1}) appropriately.

In finite dimensions it can be shown that exact controllability is equivalent to the Lie algebra rank condition (LARC) which says that
\begin{equation} \label{eq:7}
  \langle iH_0, \dots, iH_N\rangle_{\mathrm{Lie},\Bbb{R}} = \mathfrak{u}(\mathcal{H})
\end{equation}
should hold. Here the left hand side denotes the Lie algebra generated by the given operators. In infinite dimensions exact controllability can
usually not be achieved \cite{ball1982controllability}. Therefore we have to look for approximate concepts, and with the definition of strong operator controllability we have already cared about
that. What is more difficult to find is an appropriate replacement for Eq. (\ref{eq:7}).
If all the $H_0, \dots, H_N$ are bounded, we can replace $\langle iH_0, \dots, iH_N\rangle_{\mathrm{Lie},\Bbb{R}}$ with the smallest strongly closed Lie subalgebra of $\mathfrak{u}(\mathcal{H})$
containing all the $H_0, \dots, H_N$. This leads to an approximate version of LARC which implies $\mathcal{G} = \mathrm{U}(\mathcal{H})$. If we also
have $\mathcal{R} = \mathcal{G}$ the proof is completed. However, there are two major problems with this reasoning: 
\begin{itemize}[leftmargin=*]
\item \textbf{Problem 1:} If some of the $H_j$ are unbounded, Eq. (\ref{eq:7}) can never hold. Even if the Lie algebra $\langle iH_0, \dots,
  iH_N\rangle_{\mathrm{Lie},\Bbb{R}}$ is well defined (this requires a dense invariant domain for all $H_0, \dots, H_N$) it is not a subalgebra of
  $\mathfrak{u}(\mathcal{H})$, since the latter only contains bounded operators. Furthermore, there is no obvious replacement for the right hand side of (\ref{eq:7}), which could repair this
  defect. For a universal condition which is applicable to all possible choices of drift and control Hamiltonians we would need a Lie algebra which contains all anti-selfadjoint, bounded or unbounded
  operators. But the corresponding set is not even a vector space. 
  One way to solve this problem is to replace in $\langle iH_0, \dots, iH_N\rangle_{\mathrm{Lie},\Bbb{R}}$ the unbounded $H_j$ with bounded ``replacement
  generators''. we might need infinitely many (or even uncountable many) of them -- like the set of spectral projections of the unbounded $H_j$.
\item \textbf{Problem 2:} The relation $\mathcal{R} = \mathcal{G}$ is in general wrong. If $\mathcal{R}$ is a proper subset of $\mathcal{G}$ (or if this
  is at least a possibility we can not exclude) knowledge of $\mathcal{G}$ is useless for a controllability proof. A typical example where this
  pathological behavior occurs, is $N=0$ (ie. wee only look at a drift) together with the generator of translations $V_t: \psi \mapsto
  \psi(\,\cdot\, - t)$ in $\mathrm{L}^2(\Bbb{R})$ as the drift $H_0$. The only solution (if we want to follow strategies which are based on an
  analysis of $\mathfrak{g}$) is to find conditions which guarantee that $\mathcal{R} = \mathcal{G}$ still holds.
\end{itemize}
In the following we will address both problems for the case where all the control Hamiltonians (i.e. $H_j$ with $j > 0$) are bounded, and the spectrum
of $H_0$ (which is unbounded) only consists of eigenvalues (which are not necessarily isolated). This is done for Problem 1 in
Sect. \ref{sec:maximal-torus} and for Problem 2 in Sect. \ref{sec:recurrence}. The section afterwards demonstrate the usage of the new scheme with a
particular example. Note that the assumptions just given are on the one hand quite restrictive, but already comprises lots of physical models from
quantum optics and atomic physics (e.g. bound systems with controlled, bounded potentials). Therefore the results we are going to present are of
relevance not only mathematically but also in a physical context.

\section{Recurrence}
\label{sec:recurrence}

In finite dimension a one parameter unitary group $\exp(itK)$ with selfadjoint generator $K \in \mathcal{B}(\mathcal{H})$ always revisits its own past
-- either exactly (if the group is periodic) or at least approximately. The latter means that for all $t_- < 0$ and all $\epsilon > 0$ there is a $t_+
> 0$ with
\begin{equation} \label{eq:8}
  \| \exp(it_-K) - \exp(it_+K)\| < \epsilon \, .
\end{equation}
This statement can be easily proved from two simple facts: 1. If the eigenvalues $\lambda_k$, $k=1,\dots,\dim(\mathcal{H})$ of $K$ are of the form
$\lambda_k = 2\pi q_k$ with $q_k \in \Bbb{Q}$ the group $\exp(itK)$ is periodic. 2. In the general case the eigenvalues can be approximated with
arbitrary precision by numbers of the given rational form. If we go to infinite dimensions this statement can be generalized without big effort if the
spectrum of $K$ consists only of eigenvalues. We only have to replace the norm approximation from Eq. (\ref{eq:8}) by a strong approximation. Recall
in this context that a strong neighborhood base of a unitary $V$ in $\mathrm{U}(\mathcal{H})$ is given by
\begin{equation} \label{eq:9}
  \mathcal{N}(V;\psi_1, \dots, \psi_M; \epsilon) = \{ W \in \mathrm{U}(\mathcal{H})\, | \, \|W\psi_k - V\psi_k\| < \epsilon \ \forall k=1,\dots,M\},
\end{equation}
with $\psi_k \in \mathcal{H}$, $k=1,\dots,M$, $\epsilon > 0$. Results of this form are available from several authors
(e.g. \cite{wallace2015recurrence,bliss2014quantum}). To keep this paper more self consistent, we will give a proof, nevertheless. The following is
taken from \cite{HoKe17}.  

\begin{prop} \label{prop:2}
  Consider a selfadjoint (possibly unbounded) operator $K$ on $\mathcal{H}$, which has only eigenvalues in its spectrum (not necessarily
  isolated). For all $t_- \in \Bbb{R}$, $t_- \leq 0$ and all strong neighborhoods $\mathcal{N}$ of $\exp(i t_- H_X)$ in the unitary group
  $\mathrm{U}(\mathcal{H})$ of $\mathcal{H}$, there is a time $t_+ \in \Bbb{R}$, $t_+ > 0$ with $\exp(i t_+ K) \in \mathcal{N}$.   
\end{prop}

\begin{proof}
  We can assume without loss of generality that $\mathcal{N}$ has the form given in Eq. (\ref{eq:9}). Now let us consider a complete basis $\phi_n$,
  $n \in \Bbb{N}$ of eigenvectors of $K$ with eigenvalues $\lambda_n = \langle\phi_n, K \phi_n\rangle$. Furthermore $N \in \Bbb{N}$ 
  is chosen such that $\| (\Bbb{1} - P_N) \psi_j\| \leq \epsilon/3$ holds for the projection $P_N$ onto the span of
  $\phi_1, \dots, \phi_N$. Since $P_N\mathcal{H}$ is invariant under $K$, the latter defines a one-parameter group of
  unitaries $\tilde{U}_t = P_N \exp(itK) P_N$ on $P_N\mathcal{H}$. Since $P_N\mathcal{H}$ is finite dimensional we can apply
  Eq. (\ref{eq:8}). Hence for $t_- < 0$ and $\epsilon > 0$ there a $t_+ > 0$ with $\| \tilde{U}_{t_+} - \tilde{U}_{t_-}\| < \epsilon/3$ in the
  operator norm.  Now the statement follows from  
  \begin{align}
    \|\exp(i t_+ K) \psi_j - \exp(i t_- K) \psi_j\| &\leq \|(\tilde{U}_{t_+} - \tilde{U}_{t_-}) P_N \psi_j\| + \| (\exp(it_+ K) - \exp(it_- K))(\Bbb{1} -
    P_N)\psi_j\| \\
    &\leq \frac{1}{3} + \| \exp(it_+ K) \| \| (\Bbb{1} - P_N)\psi_j\| + \| \exp(it_- K) \| \| (\Bbb{1} - P_N)\psi_j\| \leq \epsilon .
  \end{align} 
\end{proof}

\noindent Now we look at the control problem (\ref{eq:1}) with $H_0$ unbounded, and pure point spectrum, and $H_k$, $k > 0$ all bounded. We can apply the
previous proposition to $H_0$ and see that the whole one-parameter group $\exp(itH_0)$ is contained in the strongly reachable set
$\mathcal{R}$. Together with a simple Trotter argument we can show that $\mathcal{R}$ and $\mathcal{G}$ coincide under the given condition.

\begin{prop} \label{prop:3}
  Assume that the spectrum of $H_0$ is pure point and all $H_k$, $k > 0$ are bounded. Then we have $\mathcal{R} = \mathcal{G}$.
\end{prop}

\begin{proof}
  We have to show that $\exp(i t H(\mathbf{y})) \in \mathcal{R}$ for all $\mathbf{y} \in \Bbb{R}^N$ and all $t \in \Bbb{R}$. By definition we already
know that this holds for all $t > 0$, and strong continuity of the one-parameter groups $\exp\bigl(itH(\mathbf{y})\bigr)$ shows that $\Bbb{1} =
s-\lim_{t\rightarrow 0} \exp(i t H_\mathbf{y}) \in \mathcal{R}$ holds as well. Hence assume $t < 0$ holds. We choose $s > 0$ and with $\tilde{t} = t -
s$ we get $t H(\mathbf{y}) = \tilde{t} H_0 + s H(ts^{-1} \mathbf{y})$. From Prop. \ref{prop:2} we see that $\exp\bigl(i\tilde{t}H_0\bigr) \in
\mathcal{R}$ holds and $\exp\bigl(isH(ts^{-1}\mathbf{y})\bigr) \in \mathcal{R}$ follows from the definition of $\mathcal{R}$. Furthermore, since the
$H_k$, $k>0$ are bounded, all $H(\mathbf{y})$ are selfadjoint on the domain $D$ of $H_0$ (cf. the corresponding remark in
Sect. \ref{sec:setup}). Hence, we can apply Trotter's product formula, and since $\mathcal{R}$ is strongly closed (by definition) we see that $\exp(t
H(\mathbf{y})) \in \mathcal{R}$ holds. Hence $\mathcal{R} = \mathcal{G}$ as stated.
\end{proof}

\begin{kor}
  Under the assumption from Prop. \ref{prop:3} we have $\mathcal{R} = \hat{\mathcal{G}}$.
\end{kor}

\begin{proof}
  This is follows immediately from Prop. \ref{prop:1} and \ref{prop:3}.
\end{proof}

\noindent For the given assumptions on the operators $H_k$, $k=0,\dots,N$ we have solved Problem 2 from the previous section. The next topic deals with
Problem 1.

\section{The maximal torus}
\label{sec:maximal-torus}

The topic of this section concerns the question: How can be find generators for $\mathfrak{g}$ if some of the $H_j$, $j=0,\dots,N$ are unbounded? The most
natural answer to this question is: Look at the spectral representation of the unbounded $H_k$. This is exactly what we will do. The first step is the
following definition.

\begin{defi} \label{def:2}
  Consider a separable Hilbert space $\mathcal{H}$ and a selfadjoint, possibly unbounded operator $K$ with spectral measure $E: \mathfrak{B}(\Bbb{R})
  \mapsto \mathcal{B}(\mathcal{H})$, where $\mathfrak{B}(\Bbb{R})$ denotes the Borel-$\sigma$-algebra of the real line. We associate to $K$ the
  Abelian von Neumann algebra
  \begin{equation}
    \mathcal{M}(K) = \{ E(\Delta) \,|\, \Delta \in \mathfrak{B}(\Bbb{R})\}'',
  \end{equation}
  where $(\,\cdot\,)'$ denotes (as usual) the commutant and $(\,\cdot\,)''$ is the double commutant. Furthermore we define the \emph{maximal torus} of $K$ as
  \begin{equation}
    \mathcal{T}(K) = \{ U \in \mathcal{M}(K)\,|\, UU^*=U^*U=\Bbb{1}\}.
  \end{equation}
\end{defi}

Abelian von Neumann algebras $\mathcal{M}$ acting on a separable Hilbert space $\mathcal{H}$ are very well understood. Only three different cases can
occur. Look first at $\mathcal{H}_c = \mathrm{L}^2([0,1])$ and the algebra of multiplication operators
\begin{equation}
  \mathcal{M}_c = \{M_f \, | \, f \in L^\infty([0,1]) \} \subset \mathcal{B}(\mathcal{H}_c),\quad \mathcal{H}_c \ni \psi \mapsto M_f\psi = f\psi \in \mathcal{H}_c.
\end{equation}
$\mathcal{M}_c$ is an Abelian (actually a \emph{maximal} Abelian) von Neumann algebra and the characteristic functions $\chi_\Delta$ belonging to
Borel subsets of $[0,1]$ give rise to all the projections $M_{\chi_\Delta}$ in $\mathcal{M}_c$. There are no \emph{minimal} projections. Another
possible case arises if we define for $n \in \Bbb{N}$ the Hilbert space $\mathcal{H}_n = \Bbb{C}^n$ and
\begin{equation}
  \mathcal{M}_n = \{ \mathrm{diag}(\lambda_1, \dots, \lambda_n)\,|\, (\lambda_1, \dots, \lambda_n) \in \Bbb{C}^n\} \subset \mathcal{B}(\mathcal{H}_n).
\end{equation}
Operators on $\mathcal{H}_n$ are just $n \times n$ matrices and $M_n$ consists of all \emph{diagonal} $n \times n$ matrices. Minimal projections are
the one-dimensional projections $E_k$ onto the elements $e_k$, $k=1,\dots,n$ of the canonical basis -- or in other words the diagonal matrices with a
one on the $k^{\mathrm{th}}$ element on the diagonal and zeros everywhere else. The final case arises if
we set $n=\aleph_0$. We get $\mathcal{H}_{\aleph_0} = \mathrm{l}^2(\Bbb{N})$ and
\begin{equation}
  \mathcal{M}_{\aleph_0} = \{ M_a \, | \, a = (a_n)_n \in \mathrm{l}^\infty(\Bbb{N})\} \subset \mathcal{B}(\mathcal{H}_{\aleph_0}),\quad
  \mathcal{H}_{\aleph_0} \ni b=(b_n)_n \mapsto M_ab = (a_nb_n)_n \in \mathcal{M}_{\aleph_0}.
\end{equation}
Hence $\mathcal{M}_{\aleph_0}$ consists of bounded sequences $(a_n)_{n \in \Bbb{N}}$. Again, the minimal projections $E_k$ arise from the canonical
basis $e_k \in \mathcal{H}_k$, $k \in \Bbb{N}$. Therefore
\begin{equation} \label{eq:13}
  E_k = M_{e_k} \in \mathcal{M}_{\aleph_0} \quad\text{with}\quad (E_k b)_j = \delta_{kj} b_j \quad\text{for}\quad b = (b_j)_j \in \mathcal{H}_{\aleph_0}.
\end{equation}
Now it can be shown \cite[Thm. 9.4.1]{kadison1997fundamentals}
that an Abelian von Neumann algebra is *-isomorphic to: $\mathcal{M}_c$, or $\mathcal{M}_n$ for some $n \in \Bbb{N} \cup \{\aleph_0\}$, or
$\mathcal{M}_c \otimes \mathcal{M}_n$.

If we look at $\mathcal{M}(K)$ the different cases are related to different spectral properties of $K$. We have $\mathcal{M}(K) \cong \mathcal{M}_c$
if $K$ has only continuous spectrum, $\mathcal{M}(K) \cong \mathcal{M}_n$ with $n \in \Bbb{N}$ if $K$ has only a pure point spectrum consisting of
exactly $n$ different eigenvalues, and $\mathcal{M}(K) \cong \mathcal{M}_{\aleph_0}$ if $K$ has again only a pure point spectrum, but its size is
countably infinite. In this paper we will look in particular at the latter case. Hence, let us denote the eigenprojections of $K$ by $F_k$, $k\in
\Bbb{N}$ (with no particular order). The von Neumann algebra  $\mathcal{M}(K)$ can now be written as (the proof of this statement is left as an exercise):
\begin{equation} \label{eq:3}
  \mathcal{M}(K) = \{ M_a\,|\, a \in \mathrm{l}^\infty(\Bbb{N})\},\quad M_a = s-\sum_{k=1}^\infty x_k F_k \quad a = (a_k)_k \in
  \mathrm{l}^\infty(\Bbb{N}),
\end{equation}
where $s-\sum$ stands for a strongly convergent series. A *-isomorphism $\Phi$ to $\mathcal{M}_{\aleph_0}$ can be defined by $\Phi(F_k) = E_k$. By
\cite[Thm. 2.4.23]{bratteli2012operator} *-morphisms between von Neumann algebras are $\sigma$-strongly continuous. Furthermore, strong and
$\sigma$-strong topologies coincide on the unit ball of $\mathcal{B}(\mathcal{H})$ for any Hilbert space $\mathcal{H}$
\cite[Prop. 2.4.1]{bratteli2012operator}. Hence, $\Phi$ defines a group isomorphism, which is at the same time a homeomorphism from $\mathcal{T}(K)$
to the standard torus 
\begin{equation} \label{eq:10}
  \mathcal{T}_{\aleph_0} = \{ U \in \mathcal{M}_{\aleph_0}\,|\, UU^* = U^*U = \Bbb{1} \}.
\end{equation}
Therefore, all statements about $\mathcal{T}_{\aleph_0}$ which only refer to its properties as a topological group carry over automatically to all
tori $\mathcal{T}(K)$ for operators with a pure point spectrum as described above. We are in particular not interested in the dimension of the
projections $F_k$. For operators with pure point spectrum this is only a minor advantage (it does not really matter in the following, whether we will
work with the projections $E_k$ or $F_k$). If we generalize the discussion of this paragraph to general selfadjoint operators, however, this is
different. In  particular the case of continuous spectrum can be reduced to an analysis of multiplication operators on $\mathrm{L}^2([0,1])$ which is 
a substantial simplification. This applies to all closedness results about subgroups and subsemigroups of the tori $\mathcal{T}(K)$, like
Thm. \ref{thm:1} below, or the discussion of recurrence from the last section.

Our next step associates to $\mathcal{T}(K)$ a Lie algebra and an exponential map. To this end consider a general (i.e. not necessarily bounded)
sequence $\mathbf{x} = (x_k)_k \in \Bbb{R}^{\aleph_0}$. We define
\begin{equation} \label{eq:11}
  D_{\mathbf{x}} = \left\{ \psi \in \mathcal{H} \,\Big| \sum_{k=1}^\infty  |\lambda_k|^2 \|F_k \psi\|^2 \leq \infty \right\}
\end{equation}
and
\begin{equation} \label{eq:12}
  X_{\mathbf{x}} : D_{\mathbf{x}} \rightarrow \mathcal{H}, \quad \psi \mapsto \sum_{k=1}^\infty F_k \psi\ .
\end{equation}
Spectral theorem \cite[Thm VIII.6]{RESI1} shows that $X_{\mathbf{x}}$ is well defined and selfadjoint on $D_{\mathbf{x}}$. By applying Stone's theorem
\cite[Thm VIII.6]{RESI1}, we see that $\Bbb{R} \ni t \mapsto \exp(itX_{\mathbf{x}}) \in \mathcal{T}(K) \subset \mathrm{U}(\mathcal{H})$ is a strongly
continuous, unitary one parameter group, and hence a continuous one-parameter subgroup of $\mathcal{T}(K)$. From the general form of operators in
$\mathcal{M}(K)$ given in Eq. (\ref{eq:3}), we see easily that \emph{all} such subgroups of $\mathcal{T}(K)$ are of this form.
\begin{equation}
  \mathfrak{T}(K) = \{ i X_{\mathbf{x}} \, | \, \mathbf{x} \in \Bbb{R}^{\aleph_0}\},
\end{equation}
is the Lie algebra of $\mathcal{T}(K)$ and the usual exponential
\begin{equation} \label{eq:14}
  \exp(i X_{\mathbf{x}}) \psi = s-\sum_{k=1}^\infty \exp(ix_k) F_k \psi, \quad \psi \in \mathcal{H},
\end{equation}
is the corresponding exponential map. Alternatively we can equip $\mathcal{T}(K)$ with the uniform topology and get a topological group again. The
continuous one-parameter subgroups are now \emph{uniformly} continuous, and therefore they are generated by the bounded elements of
$\mathfrak{T}(K)$. We denote the corresponding subspace by $\mathfrak{t}(K)$. It can be written alternatively as
\begin{align}
  \mathfrak{t}(K) &= \{i X_{\mathbf{x}}\,|\, \mathbf{x} \in \mathrm{l}^{\infty}(\Bbb{N})\} \\
  &= \text{strong closure of}\quad \mathrm{span}_{\Bbb{R}} \{i F_k\,|\, k \in \Bbb{N} \}, \label{eq:16}
\end{align}
where $\mathrm{span}_{\Bbb{R}}$ denotes the space of finite, \emph{real} linear combinations. Now we are ready for the main result of this section. It
deals with the question: What is the closure of $\{\exp(itX_{\mathbf{x}})\,|\, t \in \Bbb{R}\}$ in $\mathcal{T}(K)$?

\begin{thm} \label{thm:1}
  Consider a selfadjoint operator $K$ and its maximal torus $\mathcal{T}(K)$ as just discussed. Assume further that the eigenvalues $x_k$, $k \in
  \Bbb{N}$ of $K$ are rationally independent (i.e. linearly independent in $\Bbb{R}$ if the latter is regarded as a vector space over $\Bbb{Q}$). Then
  the closure of $\exp(itK)$ in $\mathcal{T}(K)$ coincides with $\mathcal{T}(K)$. 
\end{thm}

\begin{proof}
  We use the fact that it is sufficient to look at the group $\mathcal{T}_{\aleph_0}$ from (\ref{eq:10}) and the operator $X_{\mathbf{x}}$ on
  $\mathcal{H}_{\aleph_0}$; cf. Eqs. (\ref{eq:11}) and (\ref{eq:12}) with $F_k$ replaced by $E_k$ from Eq. (\ref{eq:13}). In other words, we ignore
  the dimension of the $F_k$, and replace these projections with the one-dimensional $E_k$. Now define $E^{[n]} = \sum_{k=1}^n E_k$ for $n \in
  \Bbb{N}$. We get a projection onto the $n$-dimensional subspace $E^{[n]}\mathcal{H}_{\aleph_0}$. An arbitrary element $V \in \mathcal{T}_{\aleph_0}$
  is given by $V = s-\sum_{k=1}^\infty \exp(2 \pi i \lambda_k) E_k$ with $\pmb{\lambda}  = (\lambda_k)_k \in \Bbb{R}^{\aleph_0}$ and therefore it
  commutes with $E^{[n]}$ and we get a unitary $V^{[n]} =  E^{[n]} V E^{[n]}$ on $E^{[n]} \mathcal{H}_{\aleph_0}$. From $\mathbf{x} \in
  \Bbb{R}^{\aleph_0}$ we derive $\hat{\mathbf{x}} = (2\pi)^{-1} \mathbf{x}$, and since the $x_k$ is by assumption algebraically independent the same
  is true for the $\hat{x}_k$. Projecting with $E^{[n]}$ we get a selfadjoint operator $K^{[n]}$ on $E^{[n]} \mathcal{H}_{\aleph_0}$
  \begin{equation}
    K^{[n]} = E^{[n]} K E^{[n]} = \sum_{k=1}^n 2\pi \hat{x}_k E_k.
  \end{equation}
  By Kronecker's Theorem \cite[Thm 7.9]{apostol1976kronecker}, this implies that we can find for all $\delta > 0$ a real
  number $t \in \Bbb{R}$ and integers $y_1, \dots, y_n$ such that
  \begin{equation}
    |t \hat{x}_k - y_k - \lambda_k| < \delta \quad \forall i=1,\dots, n
  \end{equation}
  holds. Applying this to $\exp(i t K^{[n]})$ and $V^{[n]}$ using the continuity of the exponential map we see that for all $\epsilon > 0$ we get $t
  \in \Bbb{R}$ with
  \begin{equation} \label{eq:15}
     \| \exp(it K^{[n]}) - V^{[n]} \| \leq  \sum_{k=1}^n \left| \exp(2\pi i t \hat{x}_k ) \exp( 2 \pi i y_k) - \exp(2 \pi i \lambda_k) \right| <
     \frac{\epsilon}{2}, 
   \end{equation}
   since the $y_k \in \Bbb{Z}$ and therefore $\exp(2 \pi i y_k) = 1$. for an arbitrary $\psi \in \mathcal{H}_{\aleph_0}$ with $\|\psi=1\|$ we find $n
   \in \Bbb{N}$ such that $\|(\Bbb{1} - E^{[n]}) \psi\| < \epsilon/4$. If $t \in \Bbb{R}$ is chosen such that Eq. (\ref{eq:15}) holds we get
   \begin{align}
     \left \| \bigl(\exp(itK)  - V\bigr) \psi \right\| &\leq \left\| \bigl(\exp(itK)  - V\bigr) E^{[n]} \psi \right \| + \left\| \bigl(\exp(itK)  -
                                                         V\bigr) \bigl(\Bbb{1} - E^{[n]}\bigr) \psi \right \|   \\
                                                       &\leq \left\| \bigl(\exp(itK^{[n]})  - V^{[n]}\bigr) E^{[n]} \psi \right \| + 2 \|\bigl(\Bbb{1}
                                                         - E^{[n]}\bigr) \psi\| < \frac{\epsilon}{2} + 2 \frac{\epsilon}{4} = \epsilon.
   \end{align}
   Hence, $V$ is in the strong closure of the group $\{ \exp(itK)\,|\, t \in \Bbb{R}\}$. Since $V \in \mathcal{T}_{\aleph_0}$ was arbitrary, this
   concludes the proof.
 \end{proof}

 We can now come back to Problem 1 from Sec. \ref{sec:setup}. If the drift Hamiltonian $H_0$ satisfies the assumptions from Thm. \ref{thm:1}, we see 
 from the definition of the dynamical group $\mathcal{G}$ that $\mathcal{T}(H_0) \subset \mathcal{G}$  holds, and therefore $\mathfrak{t}(H_0) \subset
 \mathfrak{g}$. Using Eq. (\ref{eq:15}) this implies $F_k \in \mathfrak{g}$. Hence, if in addition all the $H_j$ with $j>0$ are bounded we can conclude
 from (here $\overline{(\,\cdot\,)}^s$ denotes the strong closure)
 \begin{equation} \label{eq:18}
   \overline{\langle F_k, H_1, \dots H_N; k \in \Bbb{N} \rangle_{\mathrm{Lie},\Bbb{R}}}^s = \mathfrak{u}(\mathcal{H})
 \end{equation}
 that $\mathfrak{g} = \mathfrak{u}(\mathcal{H})$ is true. But this implies $\mathcal{G}=\mathrm{U}(\mathcal{H})$, and strong controllability follows
 from Prop. \ref{prop:3}. In the next Section we are looking at a special case, where this particular chain of arguments can be applied. 

\section{A controllability theorem}
\label{sec:contr-theor}

To turn the ideas from the end of the last section into a controllability proof, we need additional conditions on the control Hamiltonians. This is 
done in terms of a complete orthonormal system $\phi_k \in \mathcal{H}$, $k \in \Bbb{N}$. We associate to the set of operators $\{H_1, \dots, H_N\}$ a
graph $\Gamma$ with vertices $\mathrm{Vert}(\Gamma) = \Bbb{N}$ and edges 
\begin{equation} \label{eq:17}
  \mathrm{Edge}(\Gamma) = \{ v,w) \in \Bbb{N}\,|\, \exists j \in \{1, \dots, N\} \ \text{with}\ \langle\phi_v, H_l \phi_w\rangle \neq 0\}.
\end{equation}
A graph $\Gamma$ is called \emph{connected}, if for all pairs of vertices $v, \omega \in \mathrm{Vert}(\Gamma)$, $v \neq w$ there is a path which
connects $v$ with $w$; cf. \cite{diestel2006graph}. We use this construction to define (taken from \cite{boscain2012weak}):
\begin{defi}
  We say that a finite set $\mathcal{F} \subset \mathcal{B}(\mathcal{H})$ of bounded operators on a separable Hilbert $\mathcal{H}$ is connected with
  respect to a complete orthonormal system $\psi_k \in \mathcal{H}$, $k \in \Bbb{N}$ of $\mathcal{H}$, if the graph $\Gamma$ defined by
  $\mathrm{Vert}(\eta) = \Bbb{N}$ and Eq. (\ref{eq:17}) is connected.
\end{defi}
Now we can apply the methods introduced in the previous two sections. The only additional assumption we are using is the non-degeneracy
of the eigenvalues $H_0$. This condition can be removed easily, but then we need more conditions on the control Hamiltonians.

\begin{thm} \label{thm:2}
  Consider a separable Hilbert space $\mathcal{H}$ and the selfadjoint operators $H_0, \dots, H_N$. Furthermore assume:
  \begin{enumerate}
  \item
    $H_0$ can be bounded or unbounded, but has only pure point spectrum. The eigenvalues $x_k$, $k \in \Bbb{N}$ are non-degenerate and rationally
    independent.
  \item
    The operators $H_1, \dots, H_N$ are bounded and the set $\{H_1, \dots, H_N\}$ is connected with respect to the complete set of eigenvectors
    $\phi_k \in \mathcal{H}$, $k \in \Bbb{N}$ of $H_0$.
  \end{enumerate}
  The control system from Eq. (\ref{eq:1}) is strongly operator controllable.
\end{thm}

\begin{proof}
  We can apply Prop. \ref{prop:3} to see that under the given assumptions $\mathcal{G} = \mathcal{R}$ holds. Hence it is sufficient to show that
  $\mathcal{G} = \mathrm{U}(\mathcal{H})$ is true, which in turn follows from $\mathfrak{g} = \mathfrak{u}(\mathcal{H})$. Now we use Thm. \ref{thm:1}
  and see that $\mathfrak{t}(H_0) \subset \mathfrak{g}$ is satisfied. Hence it is sufficient to check the validity of Eq. (\ref{eq:17}) from the end
  of the last section. To make the calculations simpler let us pass temporarily to complex Lie algebras. In other words: instead of Eq. (\ref{eq:18})
  we prove
  \begin{equation} \label{eq:20}
    \overline{\langle F_k, H_1, \dots H_N; k \in \Bbb{N} \rangle_{\mathrm{Lie},\Bbb{C}}}^s = \mathcal{B}(\mathcal{H}),
  \end{equation}
  and recover Eq. (\ref{eq:18}) by a restriction to anti-selfadjoint elements. By assumption the eigenvalues $x_k$ of $H_0$ are non-degenerate. Hence
  the $F_k$ are one dimensional, and project onto the (normalized) eigenvectors $\psi_k$. Using the usual ``ketbra'' notation from quantum physics,
  this can be written as $F_k = \kb{\phi_k}$. Now assume
  \begin{equation}
    \alpha = \langle\phi_v, H_l \phi_w\rangle = \overline{ \langle\phi_v, H_l \phi_w\rangle} \neq 0
  \end{equation}
  for some $v,w = 1,\dots, N$ and $v \neq w$. We can easily compute the double commutators $[F_v,[H_l,F_w]]$ and get
  \begin{equation}
    [F_w,[H_l,F_v]] = \alpha (\KB{\phi_v}{\phi_w} + \overline{\alpha} \KB{\phi_w}{\phi_v}.
  \end{equation}
  Taking another commutator with $F_v$ leads to
  \begin{equation}
    [F_v, [F_w,[H_l,F_v]]] = \alpha (\KB{\phi_v}{\phi_w} - \overline{\alpha} \KB{\phi_w}{\phi_v}.
  \end{equation}
  Taking linear combinations we get
  \begin{equation} \label{eq:19}
    \KB{\phi_v}{\phi_w}, \KB{\phi_v}{\phi_w} \in \langle F_k, H_1, \dots H_N; k \in \Bbb{N} \rangle_{\mathrm{Lie},\Bbb{C}}
  \end{equation}
  Now we use the assumption on connectedness of the set $\{H_1, \dots, H_N\}$. This implies that we can find for any pair $v,w \in \Bbb{N}$, $v\neq w$
  a sequence $(k_1,\dots,k_M) \in \Bbb{N}^M$ with $v=k_1$, $w=k_M$ and
  \begin{equation}
    \forall j \in \{1,\dots,M-1\} \ \exists l \in \{1, \dots, N\} \quad\text{with}\quad \langle \phi_{k_j}, H_l \phi_{k_{j+1}}\rangle \neq 0 \, .
  \end{equation}
  With Eq. (\ref{eq:19}) we see that
  \begin{equation}
    \KB{\phi_{k_j}}{\phi_{k_{j+1}}} \in \langle F_k, H_1, \dots H_N; k \in \Bbb{N} \rangle_{\mathrm{Lie},\Bbb{C}} \quad \forall j=1,\dots, M-1
  \end{equation}
  holds. Taking more commutators leads to
  \begin{equation}
    \KB{\phi_v}{\phi_w} = \bigl[\KB{\phi_{k_1}}{\phi_{k_{2}}}, \bigl[ \KB{\phi_{k_2}}{\phi_{k_{3}}}, \dots,  \KB{\phi_{k_{M-1}}}{\phi_{k_{M}}}\bigr] \dots
    \bigr] \in \langle F_k, H_1, \dots H_N; k \in \Bbb{N} \rangle_{\mathrm{Lie},\Bbb{C}}\, .
  \end{equation}
  Since $\kb{\phi_v} = F_v \in \langle F_k, H_1, \dots H_N; k \in \Bbb{N} \rangle_{\mathrm{Lie},\Bbb{C}}$ this implies
  \begin{equation}
    \KB{\phi_v}{\phi_w} \in \langle F_k, H_1, \dots H_N; k \in \Bbb{N} \rangle_{\mathrm{Lie},\Bbb{C}} \quad \forall (v,w) \in \Bbb{N}^2.
  \end{equation}
  The strong closure of the span of the $\KB{\phi_v}{\phi_w}$ coincide with $\mathcal{B}(\mathcal{H})$. Hence, we get Eq. (\ref{eq:20}) and with the
  above reasoning, the statement is proved. 
\end{proof}

This theorem recovers older result from \cite{boscain2012weak}, but with an easier and more transparent proof. Please note that more general result
are in the meantime available \cite{boscain2014multi,caponigro2018exact} which do not fit into the path paved by Thm. \ref{thm:1}; i.e. the assumption
on rationally independent eigenvalues are not met. This is not necessarily a problem for the Lie algebraic methods proposed in this paper, because a
number of possible generalizations are available. This is the topic of the next section.

\section{Generalizations}
\label{sec:generalizations}

Let us add some remarks how we can proceed if the assumptions from Thm. \ref{thm:2} are not satisfied. The most easily handled generalization arises if the
eigenvalues of $H_0$ are still rationally independent but no longer non-degenerate. In that case the projections $F_k$ are not one-dimensional, but
Thm. \ref{thm:1} sill holds. Therefore we can proceed as in the proof of Thm. \ref{thm:2} and calculate commutators $[F_w,[H_l,F_v]]$ and $[F_v,
[F_w,[H_l,F_v]]]$ to see that all operators $F_v H_l F_w$ for $v,w \in \Bbb{N}$ and $l=1,\dots,N$ are elements of $\mathfrak{g}$. Strong controllability arises iff the
smallest Lie subalgebra of $\mathfrak{u}(\mathcal{H})$ containing all these operators is all of $\mathfrak{u}(\mathcal{H})$. This can still be
difficult to check, but it is at least a straightforward strategy for a proof. If the degeneracies of the eigenvalues are finite, and the projections
$F_k$ therefore finite dimensional, many methods from finite dimensions become applicable and the controllability proof therefore turns into a
family of finite dimensional problems.

If the eigenvalues of $H_0$ fail to be independent the situation becomes more difficult. We have to look at the strongly closed subgroup
$\mathcal{G}(H_0)$ of $\mathcal{T}(H_0)$ which is generated by the one parameter group $t \mapsto \exp(itH_0)$. This group defines its own Lie
algebra $\mathfrak{g}(H_0) \subset \mathfrak{t}(\mathcal{H}_0)$, and if $\mathfrak{g}(H_0)$ contains enough elements, we can proceed effectively as
in the degenerate spectrum case and calculate commutators of the form $[A,[H_l,B]]$ with $l=1,\dots,N$ and $A,B \in \mathfrak{g}(H_0)$. We get a
subset of $\mathfrak{g}(H_0,\dots,H_N)$ which should generate all of $\mathfrak{u}(\mathcal{H})$ to ensure strong controllability. If
$\mathfrak{g}(H_0)$ contains enough finite rank elements, the controllability proof can again be reduced to a family of finite dimensional problems.

To understand which drift Hamiltonians can be treated with the strategy outlined in the last paragraph we need a detailed analysis of the closed
subgroups of the standard torus $\mathcal{T}_{\aleph_0}$. It is, however, clear that the Lie algebra $\mathfrak{g}(H_0)$ can be trivial (please check
yourself that the Hamiltonian of the harmonic oscillator belongs into this class). A possible way to proceed in that case is to look for (non-Abelian) von
Neumann algebras generated by more than one operator. This is technically more challenging but has the advantage that systems with more than one
unbounded operator are included. Hence assume $H_0, \dots, H_M$ with $M < N$ are unbounded, and $H_{M+1}, \dots, H_K$, $M < K < N$ are considered as
additional (bounded) generators for the von Neumann algebra. We define
\begin{equation}
  \mathcal{A}%
  = \{ E_0(\Delta_0), \dots, E_M(\Delta_M), H_{M+1}, \dots, H_K\,|\,\Delta_1, \dots, \Delta_M \in
  \mathfrak{B}(\Bbb{R})\}'' \, ,
\end{equation}
where $E_k(\Delta_k)$ denote the spectral projections of the selfadjoint operators $H_j$. As in the Abelian case our first question should be: What
can classification theory can tell us about $\mathcal{A}$? At the coarsest level we can look at the type. The type I case is most easy to discuss and
should therefore be considered first. Most (if not all) potential applications in quantum mechanics and quantum optics fits into this class. We are
in particular interested into candidates with a big center $\mathcal{C} = \mathcal{A} \cap \mathcal{A}'$. The latter is an Abelian von Neumann algebra
and the classification of Sect. \ref{sec:maximal-torus} applies. Let us assume that $\mathcal{C} = \mathcal{M}_{n}$ holds with $n \in \Bbb{N} \cup
\{\aleph_0\}$ and the notation introduced in Sect. \ref{sec:maximal-torus} (this fits into our focus on operators with pure point spectrum only). In that case $\mathcal{M}$ is unitarily equivalent to
a direct sum  
\begin{equation} \label{eq:21}
  \mathcal{M} \cong \bigoplus_{k=1}^n \mathcal{B}(\mathcal{H}_k) \otimes \Bbb{1}_k,\quad \mathcal{H} \cong \bigoplus_{k=1}^n \mathcal{H}_k
  \otimes \mathcal{K}_k,
\end{equation}
where $\Bbb{1}_k$ denotes the unit operator on $\mathcal{K}_k$; cf. \cite[Sects. 9.3 and 9.4] {kadison1997fundamentals}. Furthermore, $\mathcal{M}$ is
*-isomorphic to the direct sum of the $\mathcal{B}(\mathcal{H}_k)$. Hence, with the same reasoning as in Sect. \ref{sec:maximal-torus} we see that the
group of unitary elements in $\mathcal{A}$:
\begin{equation}
  \mathrm{U}(\mathcal{A}) = \{ U \in \mathcal{A}\,|\, UU^* = U^*U = \Bbb{1} \}
\end{equation}
is isomorphic as a topological group (equipped with the strong topology) to a direct product of unitary groups $\mathrm{U}(\mathcal{H}_j)$.
We can associate to $\mathrm{U}(\mathcal{A})$ in the usual way a Lie algebra $\mathfrak{u}(\mathcal{A})$ consisting of anti-selfadjoint elements of
$\mathcal{A}$. As the corresponding group, we can look at $\mathfrak{u}(A)$ as the direct sum of the $\mathfrak{u}(\mathcal{H}_k)$. Our goal is again
to find enough elements in $\mathfrak{g}(H_0, \dots, H_K) \subset \mathfrak{u}(\mathcal{A})$, such that the commutators with the remaining controls
$H_{K+1}, \dots, H_N$ yields (in a strong approximation) all of $\mathfrak{u}(\mathcal{H})$.

The tricky part of this program is the determination of $\mathfrak{g}(H_0, \dots, H_K)$. This is most easily done if all the subspaces
$\mathcal{H}_k$ are one-dimensional, but this is already covered by the Abelian discussion above. Hence the next best case arises with all the
$\mathcal{H}_k$ finite dimensional (and  Eq. (\ref{eq:21}) holds with $n = \infty$). To keep the structure of $\mathcal{A}$ easy to handle, we should
not add too many generators. There is no choice for the unbounded Hamiltonians $H_0, \dots, H_M$. But the bounded generators $H_{M+1}, \dots, H_K$
are only needed to guarantee that $\mathfrak{g}(H_0, \dots, H_K)$ contains enough elements. Hence, their number should be as small as possible such that the generated algebra has an easy structure
(like all $\mathcal{H}_j$ finite dimensional). 

A system where the program just sketched can be implemented easily is a two-level atom interacting with a cavity (described by a harmonic oscillator via a quadratic interaction 
term). This is well known as the ``Jaynes-Cummings model'' in quantum optics \cite{JC63}. A control theoretic analysis is available in a large number
of papers \cite{brockett2003controllability,rangan2004control,yuan2007controllability,bloch2010finite,yuan2007controllability,KZSH}. Our discussion
will follow in particular \cite{KZSH}. We describe the model by the Hilbert space $\mathcal{H} = \Bbb{C}^2 \otimes \mathrm{L}^2(\Bbb{R})$ and in
$\mathcal{H}$ we use the complete orthonormal system $\ket{j} \otimes \ket{m}$ where $\ket{0}$, $\ket{1} \in \Bbb{C}^2$ is the canonical basis and
$\ket{m} \in \mathrm{L}^2(\Bbb{R})$, $m \in \Bbb{N}_0$ are the Hermite functions. Often we relabel it as
\begin{equation}
  \ket{\mu; 0} = \ket{0} \otimes \ket{\mu},\quad \text{and if $\mu > 0$}\quad \ket{\mu;1} = \ket{1} \otimes \ket{\mu-1}.
\end{equation}
With the Pauli operators $\sigma_\alpha$, $\alpha=\pm,1,\dots,3$ and the ordinary creation and annihilation operators $a^*,a$ on
$\mathcal{L}^2(\Bbb{R})$ we can define the drift Hamiltonian as
\begin{equation} \label{eq:22}
  H_0 = \omega_A \sigma_3 \otimes \Bbb{1} + \omega_C \Bbb{1} \otimes a^*a + \omega_I \bigl( \sigma_+ \otimes a + \sigma_- \otimes a^*),
\end{equation}
where $\omega_A, \omega_C$ and $\omega_I$ are non-vanishing but otherwise arbitrary constants. $H_0$ is well defined and essentially selfadjoint on
the domain $D_0$ consisting of all finite linear combinations of basis vectors $\ket{\mu;j}$. The control Hamiltonians are
\begin{equation} \label{eq:23}
  H_1 = \sigma_3 \otimes \Bbb{1},  \quad H_2 = \sigma_1 \otimes \Bbb{1}.
\end{equation}
The von Neumann algebra $\mathcal{A}$ which is generated by the bounded operator $H_1$ and the spectral projections of $H_0$ can be written in terms of the subspaces 
\begin{equation}
  \mathcal{H}^{(\mu)} = \mathrm{span} \{ \ket{\mu;0}, \ket{\mu;1} \} \quad \text{if $\mu > 0$, and}\quad \mathcal{H}^{(0)} = \Bbb{C} \ket{\mu;0} \quad
  \text{if $\mu=0$ holds.}
\end{equation}
Obviously $\mathcal{H}$ coincides with the direct sum of the $\mathcal{H}^{(\mu)}$, which are at the same time invariant subspaces of the operators
$H_0$, $H_1$. It is easy to see that the von Neumann algebra $\mathcal{M}$ generated by these two Hamiltonians is given by
\begin{equation}
  \mathcal{M} = \bigoplus_{\mu=0}^\infty \mathcal{B}(\mathcal{\mathcal{H}^{(\mu)}},
\end{equation}
and with a little bit more effort we can also show \cite{KZSH} that
\begin{equation}
  \mathfrak{g}(H_0,H_1) = \bigoplus_{\mu=1}^\infty \mathfrak{su}(\mathcal{H}^{(\mu)})
\end{equation}
holds, where $\mathfrak{su}(\mathcal{H}^{(\mu)})$ denotes the Lie algebra of all trace-free, anti-selfadjoint operators on $\mathcal{H}^{(\mu)}$.
Hence, since all $\mathcal{H}^{(\mu)}$ with $\mu < 0$ are two dimensional we get an infinite direct sum of $\mathfrak{su}(2)$ Lie algebras. A useful set of generators 
consists of the operators $\sigma^{(\mu)}_\alpha$, $\alpha = 1,\dots,3$, $\mu \in \Bbb{N}$, which coincide with Pauli matrices $\sigma_\alpha$ on
$\mathcal{H}^{(\mu)}$ (if we identify $\mathcal{H}^{(\mu)}$ with $\Bbb{C}^2$ via the isomorphism $\Bbb{C}^2 \ni \ket{j} \mapsto \ket{\mu;j} \in
\mathcal{H}^{(\mu)}$) and are zero everywhere else (i.e. on all $\mathcal{H}^{\nu}$ with $\nu \neq \mu$). Now, strong controllability of the control
system (\ref{eq:1}) with $H_0$, $H_1$ and $H_2$ from Eqs. (\ref{eq:22}) and (\ref{eq:23}) can be easily proved in terms of commutators between the
$\sigma^{(\mu)}_\alpha$ and the one remaining control Hamiltonian $H_2$. The corresponding calculations are very similar to those given in the proof of Thm. \ref{thm:2}. The only difference is that
all generators are traceless and therefore, the projections onto the subspaces $\mathcal{H}^{(\mu)}$ can not be generated algebraically. This problem can be solved by looking at operators of the form
$\kb{\mu;j} - \kb{\nu;j}$ and sending $\nu$ to infinity. The corresponding sequence converges strongly to $\kb{\mu;j}$. Hence the latter has to be an element of the strongly closed Lie algebra
$\mathfrak{g}$. To work out the details is left as an exercise to the reader.

\section{Outlook}
\label{sec:outlook}

Within our assumption about the Hamiltonians $H_0, \cdots, H_N$, we have seen that an approximate version of LARC can be used as a test for controllability of infinite dimensional control
problems. Compared to finite dimensions, two new, technical tasks arise: Firstly we have to calculate strong closures. This is often quite easy, as we have seen in the proof of
Thm. \ref{thm:2}. Secondly, we need a detailed spectral analysis of the drift Hamiltonian $H_0$. This is much harder and even if all eigenvalues of $H_0$ can be determined, it might happen that the
strongly closed group $\mathcal{G}(H_0)$ generated by the $\exp(itH_0)$ does not admit a non-trivial Lie algebra and in that case the scheme developed so far fails. Nevertheless, even in its current
shape the proposed methods provide already powerful tools to handle controllability problems, which complement other techniques like Galerkin approximations
\cite{boscain2012weak,boscain2014multi,caponigro2018exact} and the other methods mentioned in the introduction. Furthermore, the proposed procedure can be generalized to cover more cases. Some of the
more straightforward ideas are already sketched in Section \ref{sec:generalizations}. Along that lines a rather comprehensive treatment of systems with a pure point spectrum drift will be
possible. More challenging is of course the continuous spectrum case. Here the recurrence arguments from Sect. \ref{sec:recurrence} do not apply and systems can occur where the reachable set
$\mathcal{R}$ is not a group. The study of the dynamical group $\mathcal{G}$ and its Lie algebra is of limited use in that case, and new ideas has to be developed. This should, however, not be
considered as a hurdle. If we take the challenge we can expect insight into qualitative new behavior of quantum control systems, which is not possible in finite dimensions.

\section*{Acknowledgments}

I would like to thank all the people who helped with this research and with the writing of the manuscript. This includes in particular: G. Dirr and T. Schulte-Herbrüggen for many useful discussions
and a careful reading of this document; the organizers of the rQUACO meeting in Besançon, September 24-26, 2018 for inviting me to this nice workshop; the participants of it for more discussions about 
the topic of this paper; and  last but not least K. H. Neeb for providing the information in \cite{neebPC}.
\printbibliography

\end{document}